%% file: matroid_center.tex
\documentclass[11pt]{article}
\usepackage[margin=1in]{geometry}

\usepackage[english]{babel}
\usepackage[utf8]{inputenc}

\usepackage[T1]{fontenc}
\usepackage{kpfonts}
\usepackage{libertine}
\usepackage[scaled=0.85]{beramono}

\usepackage{graphicx}
\usepackage[colorinlistoftodos]{todonotes}
\usepackage{hyperref}
\usepackage{amsthm}
\usepackage{amsmath}
\usepackage{amssymb}
\usepackage{amsfonts}
\usepackage{mathrsfs}
\usepackage{mathtools}
\usepackage{verbatim}
\usepackage{footnote}
\usepackage{algorithm}
\usepackage[noend]{algpseudocode}
\usepackage{lineno}
\usepackage{caption}
\usepackage{tikzsymbols}
\usepackage[capitalize, nameinlink]{cleveref}
\hypersetup{colorlinks={true},linkcolor={blue},citecolor=cyan}
\usepackage{tikz}
\usetikzlibrary{arrows.meta}
\usepackage{boxedminipage2e}

\usepackage{eqparbox,array}

\input{header}

\begin{document}

\title{Small Space Stream Summary for Matroid Center
} \author{
  Sagar Kale \\
  EPFL\\
   \texttt{sagar.kale@epfl.ch}
 }
\date{}

\maketitle

\input{abstract}

\newpage

\input{intro}
\input{prelim}
\input{lb}
\input{algorithm}
\input{outliers_algorithm}
\input{guessing}
\input{ack}
\bibliographystyle{alpha}
\bibliography{ref}

\appendix
\input{offline}

\end{document}

%% file: header.tex
\crefname{theorem}{Theorem}{Theorems}
\Crefname{lemma}{Lemma}{Lemmas}
\Crefname{claim}{Claim}{Claims}
\Crefname{observation}{Observation}{Observations}

\newtheorem{theorem}{Theorem}
\newtheorem{lemma}[theorem]{Lemma}

\newtheorem*{remark}{Remark}

\renewcommand{\ge}{\geqslant}
\renewcommand{\le}{\leqslant}

\renewcommand{\leq}{\leqslant}

\newcommand{\eps}{\varepsilon}

\DeclareMathOperator{\rank}{rank}
\DeclareMathOperator{\spn}{span}

\DeclareMathOperator{\Mark}{Marker}
\DeclareMathOperator{\Cov}{Coverer}

\newcommand{\opt}{\operatorname{OPT}}

\newcommand{\indx}{\textsc{index}\xspace}
\newcommand{\bB}{\mathfrak{B}}
\newcommand{\cM}{\mathcal{M}}
\newcommand{\cI}{\mathcal{I}}
\newcommand{\cJ}{\mathcal{J}}

\newcommand{\inst}{\mathbb{I}}

\newcommand{\guess}{\tau}

\newcommand{\offguess}{\alpha}

\newcommand{\rlb}{q}

\newcommand{\lobebe}{\log(1/\eps)/\eps}
\newcommand{\ldbe}{(\log\Delta)/\eps}

\newcommand{\RR}{\mathbb{R}}


%% file: abstract.tex
\begin{abstract}
  In the matroid center problem, which generalizes the $k$-center problem, we
  need to pick a set of centers that is an independent set of a matroid with
  rank $r$.  We study this problem in streaming, where elements of the ground
  set arrive in the stream.  We first show that any randomized one-pass
  streaming algorithm that computes a better than $\Delta$-approximation for
  partition-matroid center must use $\Omega(r^2)$ bits of space, where $\Delta$
  is the aspect ratio of the metric and can be arbitrarily large.  This shows a
  quadratic separation between matroid center and $k$-center, for which the
  Doubling algorithm~\cite{CharikarCFM97} gives an $8$-approximation using
  $O(k)$-space and one pass.  To complement this, we give a one-pass algorithm
  for matroid center that stores at most $O(r^2\log(1/\eps)/\eps)$ points (viz.,
  stream summary) among which a $(7+\eps)$-approximate solution exists, which
  can be found by brute force, or a $(17+\eps)$-approximation can be found with
  an efficient algorithm.  If we are allowed a second pass, we can compute a
  $(3+\eps)$-approximation efficiently; this also achieves almost the known-best
  approximation ratio (of $3+\eps$) with total running time of
  $O((nr + r^{3.5})\lobebe + r^2\ldbe)$, where $n$ is the number of input
  points.

  We also consider the problem of matroid center with $z$ outliers and give a
  one-pass algorithm that outputs a set of $O((r^2+rz)\log(1/\eps)/\eps)$ points
  that contains a $(15+\eps)$-approximate solution.  Our techniques extend to
  knapsack center and knapsack center with $z$ outliers in a straightforward way,
  and we get algorithms that use space linear in the size of a largest feasible
  set (as opposed to quadratic space for matroid center).
\end{abstract}


%% file: intro.tex
\section{Introduction}
\label{sec:introduction}

In the $k$-center problem, the input is a metric, and we need to select a set of
$k$ centers that minimizes the maximum distance between a point and its nearest
center.  Matroid center is a natural generalization of $k$-center, where, along
with a metric over a set, the input also contains a matroid of rank $r$ over the
same set.  We then need to choose a set of centers that is an independent
set of the matroid that minimizes the maximum distance between a point and its
nearest center.  Then $k$-center is rank-$k$-uniform-matroid center.  Examples of clustering problems where the set of centers needs to form an independent set of a partition matroid arise in content distribution networks (see Hajiaghayi et al.~\cite{Hajiaghayi10} and references therein).  A partition matroid constraint can also be used to enforce fairness conditions such as having $k_M$ centers of type M and $k_W$ centers of type W.  As another example, say the input points lie in a euclidean space, and we are required to output linearly independent centers, then this is the linear-matroid center problem.
Studying a combinatorial optimization problem in the streaming model is
worthwhile not only in its own right, but also because it can lead to discovery of much faster
algorithms\footnote{This is demonstrated by Chakrabarti and
  Kale~\cite{ChakraKaleSubmod} who give streaming algorithms for submodular
  maximization problems that make \emph{only} $2|E|$ total submodular-oracle
  calls ($\tilde{O}(|E|)$ total time) and achieve constant-factor approximations, where $E$ is the ground set.
  On the other hand earlier fastest algorithms were greedy and potentially could make
  $\Omega(|E|^2)$ oracle calls.  Trivially, $|E|$ oracle calls are needed for
  any non-trivial approximation.}.

In the
streaming model, the input points arrive in the stream, and we are interested in
designing algorithms that use space sublinear in the input size.  We
study the matroid center problem in the streaming model.  By a clean reduction
from the \indx problem, we first show that any randomized one-pass streaming
algorithm that computes a better than $\Delta$-approximation for
matroid center must use $\Omega(r^2)$ bits of space, where $\Delta$ is
the aspect ratio of the metric (ratio of the largest distance to the smallest
distance between two points), which can be arbitrarily large.  Since the
Doubling algorithm~\cite{CharikarCFM97} gives an $8$-approximation for
$k$-center in one pass over the stream by storing at most $k$ points, we get a
quadratic separation between matroid center and $k$-center.  We then give a
one-pass algorithm that computes a $(7+\eps)$-approximation using a \emph{stream
  summary} of $O(r^2\lobebe)$ points.  The algorithm maintains an
efficiently-updatable summary, and runs a brute-force step when the end of
the stream is reached.  We can replace the brute-force step by an efficient
algorithm to get a $(17+\eps)$-approximation.  Alternatively, using a second
pass, we can (efficiently) compute a $(3+\eps)$-approximation.  Our algorithms
assume only oracle accesses to the metric and to the matroid.
Aforementioned efficient one-pass and two-pass algorithms have total running
time $O((nr + r^{3.5})\lobebe + r^2\ldbe)$, where $n$ is the number of input
points.

In $k$-center or matroid center, even very few rogue points can wreck up the
solution, which motivates the outlier versions where we can choose up to $z$
points that our solution will not serve.  McCutchen and
Khuller~\cite{McCutchenK08} give a one-pass $(4+\eps)$-approximation algorithm
for $k$-center with $z$ outliers that uses space $O(kz\lobebe)$.  Building on
their ideas, we give a $(15+\eps)$-approximation one-pass algorithm for matroid
center with $z$ outliers, using a brute-force search through the summary as the
last step, and a $(51+\eps)$-approximation algorithm if we want an efficient
implementation in the last step.

To the best of our knowledge, matroid center problems have not been considered
in streaming.  Chen, Li, Liang, and Wang~\cite{Chen2016} give an offline
$3$-approximation algorithm for matroid center and a $7$-approxi\-mation
algorithm for the outlier version; this approximation ratio is improved to $3$
by Harris et al.~\cite{Harris17}.  These algorithms are not easily adaptable to
the streaming setting if we are allowed only one pass, though, our two-pass
algorithm for matroid center may be thought of as running multiple copies of
Chen et al.'s $3$-approximation algorithm.  We mention that optimization
problems over matroid or related constraints have been studied
before in streaming~\cite{ashwin,ChakraKaleSubmod,chekuri15}.

The Doubling algorithm~\cite{CharikarCFM97} gives an $8$-approximation for
$k$-center.  Guha~\cite{Guha09}, using his technique of ``stream-strapping'',
improves this to $2+\eps$.  We use the stream-strapping technique in this paper
to reduce space-usage of our algorithms as well.  Known streaming algorithms for
$k$-center problems do not extend to the matroid center problems.  Indeed, the
gap between the space complexities of $k$-center and matroid center, exhibited
by our lower bound, warrants the need for new ideas.

\subsection*{Techniques}
\label{sec:techniques}

At the heart of many algorithms for $k$-center is Gonzalez's~\cite{Gonzalez85}
furthest point heuristic that gives a $2$-approximation.  It first chooses an
arbitrary point and adds it to the current set $C$ of centers.  Then it chooses
a point that is farthest from $C$ and adds it to $C$.  This is repeated until
$C$ has $k$ centers.  Let $C_E$ be the set of centers returned by this
algorithm, and let $p$ be the point that is farthest from $C_E$.  Then
$d(p,C_E)$ is the cost of the solution, whereas the set $C_E\cup\{p\}$ of size
$k+1$ acts as a certificate that an optimum solution must have cost at least
$d(p,C_E)/2$.  This can be easily implemented in streaming if we are given a
``guess'' $\guess$ of $\opt$, i.e., the cost of an optimum solution.  When we
see a new point $e$ in the stream, we add it to $C$ if $d(e,C) > 2\guess$.
Assuming that we know the aspect ratio $\Delta$, we can do this for
$2\log_{1+\eps}\Delta$ guesses of $\opt$ to get a $(2+\eps)$-approximation as
follows.  Let $R$ be the distance between first two points in the stream.  Then
maintain the set $C$ as described above for guesses
$\guess \in \{R/\Delta, (1+\eps)R/\Delta, (1+\eps)^2R/\Delta, \ldots,
R\Delta\}$.  The stream-strapping technique reduces the number of active guesses
to $O(\lobebe)$.

In extending this to matroid center, the biggest challenge is deciding which
point to make a center. In a solution to $k$-center, if we replace a point by
another point that is very close to it, then the cost can change only slightly,
whereas if we do the same in a solution to matroid center, the solution might
just become infeasible.
Therefore, if we maintain a set $C$ as earlier, it might quickly lose its
independence in the matroid.  The idea is to store, for each of the at most $r$
points $c\in C$, a maximal independent set $I_c$ of points close to $c$; here,
by close we mean close in terms of the guess $\guess$.  This way, we store at
most $r^2+r$ points.  Storing a maximal independent set for each point in $C$ may
seem wasteful, but our lower bound shows that it is necessary.  Our first
algorithmic insight is to show that this idea works for a correct guess.  We
show that if each optimum center $s$ is in the span of an independent set $I_c$
for a $c$ that is close to $s$, then we can recover an independent set of small
cost from the \emph{summary} $\bigcup_{c\in C_E} I_c$.  And as our second
insight, we show how to extend the stream-strapping approach to reduce the
number of active guesses, which helps us reduce the space usage.  These ideas
naturally combine with those of McCutchen and Khuller~\cite{McCutchenK08} and
help us design an algorithm for matroid center with $z$ outliers, but it is
nontrivial to prove that the combination of these ideas works.

\paragraph*{Knapsack center}

In the knapsack center problem, each point $e$ has a non-negative weight $w(e)$,
and the goal is to select a set $C$ of centers that minimizes the maximum
distance between a point and its nearest center subject to the constraint that
$\sum_{c\in C}w(c)\le B$, where $B$ is the \emph{budget}.  The $k$-center
problem is a special case with unit weights and $B=k$.  In the streaming
setting, our algorithms for matroid center and matroid center with outliers can
be extended to get constant approximations using space proportional to the size
of a largest feasible set, i.e., $\max\{|S|: \sum_{e\in S}w(e) \le B\}$.  As
described earlier, we maintain a set $C$ of potential centers using the guess
$\guess$, and for each potential center $c$, we also maintain a smallest
weight point, say $s_c$, in its vicinity.  Then, in the end, the summary
$\{s_c: c\in C\}$ contains a good solution.  This idea works because replacing
a center by a nearby point with a smaller weight does not affect the feasibility
in the knapsack setting (which could destroy independence in the matroid
setting).

\subsection*{Related Work}
\label{sec:related-work}
The $k$-center problem was considered in the '60s~\cite{Hakimi64,Hakimi65}.  It
is NP-hard to achieve a factor of better than $2$~\cite{Hsu79}, and
polynomial-time $2$-approximation algorithms exist~\cite{Gonzalez85,Hochbaum85}.
As mentioned earlier, Chen et al.~\cite{Chen2016} give a $3$-approximation
algorithm for matroid center and a $7$-approximation algorithm for the outlier
version, and this approximation ratio is improved to $3$ by Harris et
al.~\cite{Harris17}.  Motivated by applications in content distribution
networks, the matroid median problem is
considered as well~\cite{Hajiaghayi10,Krishnaswamy11}.  The problem of
$k$-center with outliers was first studied by Charikar et
al.~\cite{CharikarKMN01} who gave a $3$-approximation algorithm.  The
approximation ratio was recently improved to $2$ by Chakrabarty et
al.~\cite{ChakrabartyGK16}. We mention the work of Lattanzi et
al.~\cite{Lattanzi15} that considers hierarchical $k$-center with outliers.

For knapsack center, a
$3$-approximation was given by Hochbaum and Shmoys~\cite{Hochbaum86}.  For the
outlier version of knapsack center, very recently, Chakrabarty and
Negahbani~\cite{chakrabartyN18} gave the first non-trivial approximation (a
$3$-approximation).

\paragraph*{Streaming} Charikar et al.~\cite{Charikar03} and Guha et
al.~\cite{Guha03} consider $k$-median with and without outliers in streaming.
Guha~\cite{Guha09} gives a $(2+\eps)$-approximation one-pass algorithm for
$k$-center that uses $O(k\lobebe)$ space, and McCutchen and
Khuller~\cite{McCutchenK08} give a $(4+\eps)$-approximation one-pass algorithm
for $k$-center with $z$ outliers that uses $O(kz\lobebe)$ space.  The special
cases of $1$-center (or, the minimum enclosing ball problem) and $2$-center in
euclidean spaces have been considered~\cite{Zarrabi-ZadehM09,Kim15,HatamiZ17}
and better approximation ratios than the general $k$-center problem are known in
streaming.  Correlation clustering is studied in streaming by Ahn et
al.~\cite{Ahn15}.  Cohen-Addad et al.~\cite{cohenaddad16} give streaming
algorithms for $k$-center in the sliding windows model, where we want to
maintain a solution for only some number of the most recent points in the
stream.  Guha~\cite{Guha09} also gives a space lower bound of $\Omega(n)$ for
one-pass algorithms that give a better than $2$ approximation for (even the
special case of) $1$-center by a simple reduction from $\indx$, where $n$ is the
number of points.

\paragraph*{$k$-center in different models}
Chan et al.~\cite{Chan18} consider $k$-center in the fully dynamic
adversarial setting, where points can be added or deleted from the input, and
the goal is to always maintain a solution by processing the input updates
quickly.  Malkomes et al.~\cite{Malkomes15} study distributed $k$-center with
outliers.


\subsubsection*{Organization of the Paper}
\label{sec:organization-paper}

We define the model and the problems in~\Cref{sec:prelim}.
\Cref{sec:space-lower-bound} is on the lower bound.  In
\Cref{sec:one-pass-algorithm}, we give our important algorithmic ideas and
discuss our algorithm for matroid center, and then in
\Cref{sec:one-pass-algorithm-ouliers}, we discuss the outlier version.  In
\Cref{sec:handling-guesses}, we give the improved space bounds.


%% file: prelim.tex
\section{Preliminaries}
\label{sec:prelim}

A matroid $\cM$ is a pair $(E,\cI)$, where $E$ is a finite set and is called the
ground set of the matroid, and $\cI$ is a collection of subsets of $E$ that
satisfies the following \emph{axioms}:
\begin{enumerate}
  \item $\emptyset \in \cI$,
  \item if $J \in \cI$ and $I\subseteq J$, then $I \in \cI$, and
  \item if $I, J \in \cI$ and $|I| < |J|$, then there exists $e \in J\setminus
  I$ such that $I\cup\{e\} \in \cI$.
\end{enumerate}
If a set $A\subseteq E$ is in $\cI$, then it is called an \emph{independent} set
of the matroid $\cM$, otherwise it is called a \emph{dependent set}.  A
singleton dependent set is called a \emph{loop}.  \emph{Rank} of a set $A$,
denoted by $\rank(A)$, is the size of a maximal independent set within $A$; note
that $\rank$ is a well-defined function because of the third axiom, which is
called the \emph{exchange} axiom.  Clearly, for $A\subseteq B$,
$\rank(A)\le\rank(B)$.  Rank of a matroid is the size of a maximal independent
set within $E$.  \emph{Span} of a set $A$, denoted by $\spn(A)$, is the largest
set that contains $A$ and has the same rank as $A$ (it can be shown that such a
set is unique).  We will also use \emph{submodularity} of the rank function,
i.e., for $A,B\subseteq E$,
\begin{equation}
  \label{eq:4}
  \rank(A\cup B) + \rank(A\cap B) \le \rank(A) + \rank(B)\,.  
\end{equation}

A matroid $(E,\cI)$ is a \emph{partition} matroid if there exists a partition
$\{E_1, E_2, \ldots, E_p\}$ of $E$ and nonnegative integers
$\ell_1, \ell_2, \ldots, \ell_p$, such that
$\cI = \{ A \subseteq E : \forall i \in [p], |A\cap E_i| \le \ell_i \}$.  We say that $\ell_i$ is the \emph{capacity} of part $E_i$.  Observe that the rank of the matroid is $\sum_{i = 1}^p \ell_i$.

A metric $d$ over $E$ is a (distance) function $d:E\times E \rightarrow \RR_+$
that satisfies the following properties for all $e_1, e_2, e_3 \in E$:
\begin{enumerate}
  \item $d(e_1,e_2)=0$ if and only if $e_1=e_2$,
  \item $d(e_1,e_2) = d(e_2,e_1)$, and
  \item $d(e_1,e_3) \le d(e_1,e_2) + d(e_2,e_3)$; this property is called the
  \emph{triangle inequality}.
\end{enumerate}
We sometimes call elements in $E$ points.  For a point $e$ and a positive number
$\alpha$, the closed ball of radius $\alpha$ around $e$, denoted by
$\bB(e,\alpha)$, is the set $\{x\in E: d(e,x) \le \alpha\}$.  We overload $d$ by
defining $d(e,A) := \min_{x\in A} d(e,x)$ for $e\in E$ and $A\subseteq E$.  The
aspect ratio $\Delta$ of a metric is the ratio of the largest distance to the
smallest in the metric, i.e., $\max_{x,y}d(x,y)/\min_{x,y}d(x,y)$.

The input for the matroid center problem is a matroid $\cM = (E,\cI)$ of rank
$r$ and a metric $d$ over $E$.  The goal is to output an independent set $S$
such that its cost $\max_{e\in E} d(e,S)$ is minimized.  We are interested in
algorithms that assume oracle (or black-box) accesses to the matroid and the
metric.  The algorithm can ask the matroid oracle whether a set is independent
or not, and it can ask the metric oracle (or distance oracle) what the distance
between given two points is.  In the streaming model, elements of $E$ arrive one
by one, and we want to design an algorithm that uses small (sublinear in the
input) space.  The algorithm can query the oracles only with the elements of
$E$.  If the algorithm queries an oracle with an element not in $E$, then we say
that it \emph{fails}.  A streaming algorithm can only remember a small part of
the input, and the aforementioned restriction disallows plausible learning about
forgotten elements indirectly from oracle calls.  Also, an algorithm cannot just
enumerate elements of $E$ on the fly without looking at the stream, because it
does not know the names of the elements in advance.

The input for matroid center with $z$ outliers is also a matroid
$\cM = (E, \cI)$ and a metric $d$ over $E$, but the goal is to output an
independent set whose cost is computed with respect to $|E| - z$ closest points.
Formally, cost of a set $S$ is
$\min\{\alpha \in \RR_+: |E \setminus (\bigcup_{s \in S}\bB(s,\alpha))| \le
z\}$.

We denote by $\opt$ the cost of an optimum solution of the instance in the
context and by $n$ the number of input points, i.e., $|E|$.


%% file: lb.tex
\section{Space Lower Bound for One Pass Matroid Center}
\label{sec:space-lower-bound}
We show that $\Omega(r^2)$ space is required to achieve better than
$\Delta$-approximation for a one-pass algorithm for matroid center.  We reduce
from the communication problem of \indx.  This reduction is based on the simple
reduction for the maximum-matching-size problem: see~\Cref{fig:lb}.  In
$\indx_N$, Alice holds an $N$-bit string and Bob holds an index $I \in [N]$;
Alice sends a message to Bob, who has to determine the bit at position $I$.  It
is known that Alice has to send a message of size at least
$(1-H_2(3/4))N \ge 2N/11$ for Bob to output correctly with a success probability
of $3/4$, where $H_2$ is the binary entropy function.

\input{fig_lb}

\subsection*{Reduction from \indx to Partition-Matroid Center}
\label{sec:reduction-from-indx}
We prove the following theorem.
\begin{theorem}
  Any one-pass algorithm for partition-matroid center that outputs a better than
  $\Delta$-approximation with probability at least $3/4$ must use at least
  $r^2/24$ bits of space.
  \label{thm:lb}
\end{theorem}
\begin{proof}

  Assume, towards a contradiction, that there exists a one-pass algorithm for
  partition-matroid center that outputs a better than $\Delta$-approximation
  using at most $r^2/24$ bits of space.  Then we use it to solve the \indx
  problem.
  Given an input for \indx, Alice and Bob first construct a bipartite graph $G$
  just as described in~\Cref{fig:lb}.  Then they construct a partition-matroid
  center instance based on $G$.  Before formalizing the construction, we
  emphasize that the metric does not correspond to the graph metric given by
  $G$, but each edge in $G$ will become a point in the metric.  The vertex set
  they use is union of four sets $C_A, V_A$, each of size $\rlb$, and
  $C_B, V_B$, each of size $\rlb - 1$.  Alice constructs a subset of edges
  between $C_A$ and $V_A$ based on her $N$-bit string, so we use $N = \rlb^2$.
  We say that these edges are owned by Alice.  If the index that Bob holds
  corresponds to an edge $\{u,v\}$ with $u\in C_A$ and $v\in V_A$, he adds a
  perfect matching $M$ between $C_A\setminus \{u\}$ and $V_B$ and a perfect
  matching $M'$ between $V_A\setminus\{v\}$ and $C_B$.  The edges in $M \cup M'$
  are owned by Bob.

  To each $u \in C_A \cup C_B$, we associate a cluster $C(u)$ of at most
  $\rlb$ points in the metric that we will construct, and to each
  $v \in V_A \cup V_B$, we associate a part $P(v)$ in the partition matroid
  with capacity $1$.  Thus, rank of the matroid  $r = 2\rlb - 1$ because $|V_A \cup V_B| = 2\rlb - 1$.  By our design,
  no two clusters will intersect and no two parts will intersect, i.e.,
  $C(u)\cap C(u')=\emptyset$ for $u \neq u'$, and $P(v)\cap P(v')=\emptyset$ for
  $v \neq v'$.  The metric is as follows.  Any two points in the same cluster
  are a unit distance apart and any two points in two different clusters
  are distance $\Delta$ apart.  This trivially forms a metric, because the
  clusters are disjoint.  For each $u \in C_A$, Bob adds a point $p(u)$ in the
  cluster $C(u)$, so that it is nonempty.  Add $P' := \{p(u) : u \in C_A\}$ as a
  part in the partition matroid with capacity $0$, so no $p(u)$ can be a center.
  For each edge $\{u,v\}$ in $G$ with $u \in C_A \cup C_B$ and
  $v \in V_A \cup V_B$, whoever owns that edge adds a point $p(\{u,v\})$ that
  goes in cluster $C(u)$ and part $P(v)$.  Now, Alice runs the partition-matroid
  center algorithm on the points she constructed.  She can do this because she
  knows the metric and the part identity of each point, so she can simulate the
  distance and matroid oracles.  Note that if the algorithm expects an explicit
  description of the partition matroid, Alice can also send along with each
  point the identity of the part to which it belongs and the capacity of the
  part (which is always $1$ for her points).  She then sends the memory contents
  to Bob, who continues running the algorithm on his points and computes the
  cost of the output.  We note that Bob can also simulate the distance and
  matroid oracles.  Any point he does not own corresponds to a red edge, and
  using the identity of that edge, he can figure out the part and cluster to
  which the point belongs.

  Now we prove the correctness of the reduction.  Say Bob holds the index
  corresponding to the edge $\{u,v\}$, where $u \in C_A$ and $v \in V_A$.  If
  the index is $1$, then $\{u,v\}$ exists in the graph, then opening centers at
  points corresponding to edges in $M \cup M' \cup \{u,v\}$ satisfies the
  partition matroid constraint and also for each $u \in C_A \cup C_B$, we have a
  center opened in $C(u)$, so the cost is $1$.  Let the index be $0$. We want to
  show that there is no independent set of cost less than $\Delta$.  For a
  contradiction, assume there is such an independent set.  Now, recall that
  $p(u)$ cannot be a center, so it has to be served by some center in $C(u)$,
  otherwise the cost will be $\Delta$.  Let $p(u)$ be served by some
  $p(\{u,v'\})$ for $v'\neq v$.  Then $p(\{v',w\})$, where $\{v',w\} \in M'$,
  cannot be a center, because both $p(\{u,v'\}$ and $p(\{v',w\})$ belong to the
  part $P(v')$ with capacity $1$.  The point $p(\{v',w\})$ is the lone point in
  its cluster, and since it cannot be a center, the cost is $\Delta$.  If the
  algorithm is better than $\Delta$-approximation, then Bob can distinguish
  between these two cases, and thus, solve $\indx_N$ using communication at most
  $r^2/24 \le 4\rlb^2/24 = N/6$ bits, which is a contradiction.
\end{proof}

After seeing the lower bound, a remark is in order.  The difficulty in designing
an algorithm is as follows.  Even if we know that one center must lie in a ball
of small radius centered at a known point, we do not know which points in that
ball to store so as to recover an independent set of the matroid.


%% file: fig_lb.tex
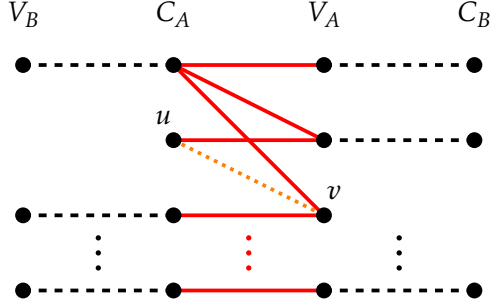
\begin{figure}[!h]
\centering
\begin{tikzpicture}[line width=0.5mm]

\draw [red] (0,0) -- (2,0);
\draw [red] (0,1) -- (2,1);
\draw [red] (0,2) -- (2,2);
\draw [dotted,orange] (0,2) -- (2,1);
\draw [red] (0,3) -- (2,1);
\draw [red] (0,3) -- (2,2);
\draw [red] (0,3) -- (2,3);

\fill [color=red] (1,0.3) circle (1pt);
\fill [color=red] (1,0.5) circle (1pt);
\fill [color=red] (1,0.7) circle (1pt);

\fill  (-1,0.3) circle (1pt);
\fill  (-1,0.5) circle (1pt);
\fill  (-1,0.7) circle (1pt);

\fill  (3,0.3) circle (1pt);
\fill  (3,0.5) circle (1pt);
\fill  (3,0.7) circle (1pt);

\draw [dashed] (2,0) -- (4,0);

\draw [dashed] (2,2) -- (4,2);
\draw [dashed] (2,3) -- (4,3);

\draw [dashed] (0,0) -- (-2,0);
\draw [dashed] (0,1) -- (-2,1);

\draw [dashed] (0,3) -- (-2,3);

\draw (-0.4,2.3) node[anchor=west] {$u$};
\draw (2.4,1.3) node[anchor=east] {$v$};

\fill [color=black] (0,0) circle (3pt);
\fill [color=black] (2,0) circle (3pt);
\fill [color=black] (0,1) circle (3pt);
\fill [color=black] (2,1) circle (3pt);
\fill [color=black] (0,2) circle (3pt);
\fill [color=black] (2,2) circle (3pt);
\fill [color=black] (0,3) circle (3pt);
\fill [color=black] (2,3) circle (3pt);

\fill [color=black] (-2,0) circle (3pt);
\fill [color=black] (4,0) circle (3pt);
\fill [color=black] (-2,1) circle (3pt);
\fill [color=black] (4,2) circle (3pt);
\fill [color=black] (-2,3) circle (3pt);
\fill [color=black] (4,3) circle (3pt);

\draw (0,4) node[anchor=north] {$C_A$};
\draw (-2,4) node[anchor=north] {$V_B$};
\draw (2,4) node[anchor=north] {$V_A$};
\draw (4,4) node[anchor=north] {$C_B$};

\end{tikzpicture}
\caption{ If we have a one-pass streaming algorithm that computes the size of a
  maximum matching of a $k$ vertex bipartite graph using $o(k^2)$ space, then we
  can solve $\indx_N$ using $o(N)$ communication, which would be a
  contradiction.  Alice and Bob agree on a bijection from $[N]$ to the edges of
  a complete bipartite graph $K_{k,k}$ and construct a graph $G$ as follows.  If
  $\ell$th bit is $1$, Alice adds the corresponding edge (shown in solid red).
  If the index corresponds to the edge $\{u,v\}$ (shown as a dotted orange
  edge), Bob adds a new perfect matching between all but vertices $u$ and $v$
  and $2k-2$ new vertices (shown as dashed black edges).  Alice runs the
  matching-size estimation algorithm and sends the memory contents to Bob, who
  continues running it and computes the output.  By design, if the index is $1$,
  then maximum-matching-size is $2k - 1$, otherwise it is $2k - 2$, and an exact
  algorithm can distinguish between the two cases.} \label{fig:lb}
\end{figure}


%% file: algorithm.tex
\section{Matroid Center}
\label{sec:one-pass-algorithm}

Our algorithm for matroid center can be seen as a generalization of the
algorithm by Hochbaum and Shmoys for $k$-center~\cite{Hochbaum85} adapted to the
streaming setting.  We first quickly describe the algorithm for $k$-center.
Given an upper bound $\guess$ on the optimum cost, the algorithm
stores a set $C$ of up to $k$ \emph{pivots} such that distance between any two
pivots is more than $2\guess$.  When the algorithm sees a new point $e$ in the
stream such that distance between $e$ and any pivot is more than $2\guess$, it
makes $e$ a pivot.  The size of $C$ cannot exceed $k$ in this way, because
$\guess$ is an upper bound on the optimum cost, so no two pivots are served by a
single optimum center.  Also, any other point is within distance $2\guess$ of
some pivot.  In the end, the algorithm designates all pivots as centers.  In
generalizing this to matroid center, one obvious issue is that the set $C$ of
pivots constructed as above may not be an independent set for the given general
matroid\footnote{This is precisely why we call points in $C$ ``pivots'' rather
  than ``centers'' in this paper.}.  What we do know is that there has to be an
optimum center within distance $\guess$ of each pivot.  Formally, for $c\in C$,
there exists $s_c$ such that $d(c,s_c)\le \guess$ and $\{s_c : c \in C\}$ is
an independent set.
For each pivot $c$, we maintain an independent set $I_c$ of nearby points.  We
prove that it is enough to have each $s_c$ be spanned by some $I_{c'}$ to get a
good solution within $\bigcup_{c\in C} I_c$.  \Cref{alg:1pmc} gives a formal
description.

Note that in~\Cref{alg:1pmc} if we try to add $e$ to $I_c$ under the condition
that $d(e,c)\le\guess$, then we may miss spanning some $s_c$.  This will happen
if $d(s_c,C) \in (\guess,2\guess]$, where $C$ is the set of pivots when $s_c$
arrived. Using the condition $d(e,c)\le\guess$ works if each $s_c$ arrives after
$c$ though (we use it in the second pass of our two-pass algorithm).

\begin{algorithm}[!h]
  \begin{algorithmic}[1]
    \Function{MatroidCenter}{$\guess$,flag}
    


    \State Initialize pivot-set $C\gets \emptyset$.

    \For{each point $e$ in the stream}

    \If{there is a pivot $c\in C$ such that $d(e,c)\le2 \guess$ (pick arbitrary
      such $c$)}

    \If{$I_c\cup\{e\}$ is independent}  \label{line:scaddif}

    \State $I_c \gets I_c\cup\{e\}$.  \label{line:scadd}

    \EndIf

    \ElsIf{$|C| = r$} \label{line:sizecheck} \Comment{We cannot have more pivots
      than the rank.}

    \State Abort. \Comment{Because $C\cup \{e\}$ acts as a certificate that the
      guess is incorrect.}

    \Else

    \State $C \gets C\cup\{e\}$.  \Comment{Make $e$ a pivot.}

    \State If $\{e\}$ is not a loop, $I_e \gets \{e\}$, else
    $I_e\gets\emptyset$.  \label{line:scaddinit}

    \EndIf

    \EndFor

    \If{flag = ``brute force''}

    \State Find an independent set $C'_B$ in $\bigcup_{c\in C}I_c$ such that
    $d(c,C'_B)\le 5\guess$ for $c \in C$.

    \State If such $C'_B$ does not exist, then abort, else return
    $C'_B$. \label{line:bruteforcecheck}

    \EndIf

    \State \Return \Call{EfficientMatroidCenter}{$5\guess$, $C$,
      $(I_c)_{c\in C}$, $\cM$} (given in~\Cref{alg:effchen}
    in~\Cref{sec:effic-appr-algor}).

    \EndFunction
    \caption{One pass algorithm for matroid center.}
    \label{alg:1pmc}
  \end{algorithmic}
\end{algorithm}

First, we quickly bound the space usage.
\begin{lemma}\label{lem:mcspace}
  In any call to \Call{MatroidCenter}{}, we store at most $r^2+r$ points.
\end{lemma}
\begin{proof}
  The check on~\Cref{line:sizecheck} ensures that $|C|\le r$.  For each pivot
  $c$, the size of its independent set $I_c$ is at most $r$, hence the total number of
  points stored is at most $r^2+r$.
\end{proof}
Consider a call to \Call{MatroidCenter}{} with $\guess\ge\opt$.  Let $C_E$ be
the set of pivots at the end of the stream.  As alluded to earlier, for an
optimum independent set $I^*$, the following holds: for each $c\in C_E$, there
exists $s_c\in I^*$ such that $d(c,s_c)\le \guess$, and also $s_c\neq s_{c'}$
for $c\neq c'$, because $d(c,c') > 2\guess$.  Now, we prove the following
structural lemma that we need later.

\begin{lemma}
  Let $I_1,\ldots,I_t$ and $S=\{s_1,\ldots,s_u\}$ be independent sets of a
  matroid such that there is an onto function $f:[u]\rightarrow [t]$ with the
  property that $s_i$ is in the span of $I_{f(i)}$ for $i\in[u]$.  Then there
  exists an independent set $B$ such that $|B\cap I_j|\ge 1$ for $j\in[t]$.
  \label{lem:matstruct}
\end{lemma}

\begin{proof}
  For each $\ell \in \{0, 1, \ldots, u\}$, we construct an independent set
  $S_\ell$ such that $|S_\ell| = u$, $|S_\ell\cap I_{f(j)}| \ge 1$ for
  $j\le \ell$, and $s_{\ell+1},\ldots,s_u \in S_\ell$, then $S_u$ is our desired
  set $B$.  Start with $S_0 = S$, and assume that we have constructed
  $S_0, S_1,\ldots,S_{\ell -1}$.  If $s_\ell\in I_{f(\ell)}$, we are done, so
  let $s_\ell\notin I_{f(\ell)}$, then we claim that
  $\rank((S_{\ell -1}\setminus\{s_\ell\})\cup I_{f(\ell)}) \ge u$.  To see this,
  observe that $\rank(S_{\ell -1}) = u$, so by monotonicity of the rank function,
  $\rank(S_{\ell -1}\cup I_{f(\ell)}) \ge u$, but
  $s_\ell \in \spn(I_{f(\ell)})$, so removing $s_\ell$ from
  $S_{\ell-1}\cup I_{f(\ell)}$ would not reduce its rank.  We now give a formal
  argument for completeness.
  We have
  $(I_{f(\ell)}\cup\{s_\ell\})\cup((S_{\ell -1}\setminus\{s_\ell\})\cup
  I_{f(\ell)})=S_{\ell -1}\cup I_{f(\ell)}$, and
  $(I_{f(\ell)}\cup\{s_\ell\})\cap((S_{\ell -1}\setminus\{s_\ell\})\cup
  I_{f(\ell)})=I_{f(\ell)}$. By submodularity of the rank function
  (see~\eqref{eq:4} in~\Cref{sec:prelim}), we have
  \[
    \rank(S_{\ell -1}\cup I_{f(\ell)}) +
    \rank(I_{f(\ell)})\leq\rank(I_{f(\ell)}\cup\{s_\ell\}) + \rank((S_{\ell
      -1}\setminus\{s_\ell\})\cup I_{f(\ell)})\,.
  \]
  Let $q=\rank(I_{f(\ell)})$.  Since $s_\ell \in \spn(I_{f(\ell)})$, we have
  $\rank(I_{f(\ell)}\cup\{s_\ell\})=q$ and the above inequality gives
  \[
    u+q\leq q+\rank((S_{\ell -1}\setminus\{s_\ell\})\cup I_{f(\ell)})\,,
  \]
  which proves the claim.
  Now,
  $\rank(S_{\ell -1}\setminus\{s_\ell\})=u-1<\rank((S_{\ell
    -1}\setminus\{s_\ell\})\cup I_{f(\ell)})$, therefore there exists
  $a\in I_{f(\ell)}$ such that
  $S_\ell := (S_{\ell -1}\setminus\{s_\ell\})\cup\{a\}$ is independent by the
  exchange axiom.
\end{proof}

\begin{lemma}[Small stream summary for matroid center]
\label{lem:existence}
  Consider a call to \Call{MatroidCenter}{} with $\guess\ge\opt$.  Then there
  exists an independent set $B\subseteq \bigcup_{c\in C_E}I_c$ such that
  $d(e,B) \le 7\guess$ for any point $e$ and $d(c,B)\le 5\guess$ for any pivot
  $c\in C_E$.
\end{lemma}
\begin{proof}
  For $c \in C_E$, denote by $s_c$ the optimum center that serves it, so
  $d(c,s_c) \le \guess$.  Let $c'\in C_E$ be such that we tried to add $s_c$ to
  $I_{c'}$ either on~\Cref{line:scadd} or on~\Cref{line:scaddinit}; note that
  $c'$ may not be the same as $c$ if we added it on~\Cref{line:scadd}.  For an
  $x \in I^*$, let $a(x)\in C_E$ denote the pivot whose independent set
  $I_{a(x)}$ we tried to add $x$ to.  Either we succeeded, in which case
  $x \in I_{a(x)}$, or we failed, in which case $x\in\spn(I_{a(x)})$.  In any
  case, by~\Cref{lem:matstruct}, for $\mathcal{A} := \{I_{a(x)} : x \in I^*\}$
  there exists an independent set $B$ such that $|I \cap B| \ge 1$ for all
  $I\in\mathcal{A}$.

  Now, we will bound the cost of $B$.  See~\Cref{fig:mccost}.  Consider any
  point $e$ in the stream.  Let
  \begin{itemize}
    \item $c(e)\in C_E$ be such that $d(e,c(e))\le 2\guess$,
    \item $s_{c(e)}$ be the optimum center that serves $c(e)$, so
    $d(c(e),s_{c(e)}) \le \guess$,
    \item $a(s_{c(e)})\in C_E$ be the pivot whose independent set we tried to
    add $s_{c(e)}$ to, so $d(s_{c(e)},a(s_{c(e)})) \le 2\guess$,
    \item $c'(e)$ be an arbitrary point in $I_{a(s_{c(e)})}\cap B$, so
    $d(a(s_{c(e)}),c'(e)) \le 2\guess$ because $c'(e)\in I_{a(s_{c(e)})}$.
  \end{itemize}
  Then by triangle inequality,
  \[
    d(e,B) \le d(e,c'(e)) \le d(e,c(e)) + d(c(e),s_{c(e)}) +
    d(s_{c(e)},a(s_{c(e)})) + d(a(s_{c(e)}),c'(e)) \le 2\guess + \guess +
    2\guess + 2\guess = 7\guess\,,
  \]
  which proves the first part of the lemma.

  For any $c\in C_E$, we can bound $d(c,B)$ in a similar way.  Let $s_c$ be the
  optimum center that serves $c$, and similarly define $a(s_c)$ to be the pivot
  such that $d(s_c,a(s_c)) \le 2\guess$.  Also, let $c'$ be the point in $B$
  such that $d(a_{s_c},c') \le 2\guess$.  This gives that
  $d(c,B) \le d(c,c') \le 5\guess$.
\end{proof}

\input{fig_mccost} Before proving our main theorem, we need the following
guarantee on the efficient offline $3$-approxi\-mation algorithm denoted by
\Call{EfficientMatroidCenter}{}.  This algorithm is based on the offline
algorithm for matroid center by Chen et al.~\cite{Chen2016}.  We give it as
input $\offguess = 5\guess$, the set $C_E$ of pivots, their independent sets
$(I_c)_{c\in C_E}$, and the underlying matroid $\cM$ with the promise on the
input that there is an independent set $B\subseteq \bigcup_{c\in C_E} I_c$ such
that for $c\in C_E$, it holds that $d(c,B)\le 5\guess = \offguess$.
\begin{theorem}
  If \Call{EfficientMatroidCenter}{} does not fail, then it outputs a set $C'$
  such that $d(c,C')\le 3\offguess$ for each $c\in C_E$.  If the input
  promise holds, then \Call{EfficientMatroidCenter}{} does not fail.
  \label{thm:chen}
\end{theorem}
\begin{proof}
  This theorem is proved as~\Cref{thm:chencopy} in the appendix.
  See~\Cref{sec:effic-appr-algor}.
\end{proof}
Now we prove the main result.
\begin{theorem}
  There is an efficient $(17+\eps)$-approximation one-pass algorithm for matroid
  center that stores at most $2(r^2+r)\log_{(1+\eps/17)} \Delta$ points and has
  total running time $O((nr + r^{3.5})\ldbe)$.  With a brute force algorithm,
  one can get a $(7+\eps)$-approximation in time $O((nr + r^{2r+2})\ldbe)$.
  \label{thm:main}
\end{theorem}
\begin{proof}
  The algorithm is as follows.  Let $\delta$ be the distance between the first two
  points.  Then for $2\log_{1+\eps/17}\Delta$ guesses $\guess$ of $\opt$
  starting from $\delta/\Delta$ to $\delta\Delta$, we run
  \Call{MatroidCenter}{$\guess$, flag}.  We return the set of centers returned
  by the instance corresponding to the smallest guess $\guess$.
  \Cref{lem:mcspace} gives the desired space bound.

  \textbf{Case 1.} flag = ``brute force''.

  Suppose the algorithm returned $C'_B$.  \Cref{lem:existence} guarantees that
  for $\guess\in[\opt,(1+\eps/17)\opt)$, the algorithm will not abort.  Then, by
  the check on~\Cref{line:bruteforcecheck}, cost of $C'_B$ is at most
  $7\guess \le (7+\eps)\opt$.  For each guess: a point in the stream is
  processed in time $O(r)$, and the postprocessing time is $O(r^{2r+2})$.

  \textbf{Case 2.} flag = ``efficient algorithm''.
  
  Let the algorithm returned $C'$.  \Cref{thm:chen} guarantees that for
  $\guess\in[\opt,(1+\eps/17)\opt)$, the algorithm will not abort.
  By~\Cref{thm:chen} for any $c\in C_E$, we have $d(c,C')\le 15\guess$.  Since
  we forget only the points within distance $2\guess$ of $C_E$, we get that for
  any point $e$ in the stream, $d(e,C')\le 17\guess\le(17+\eps)\opt$.  For each
  guess: a point in the stream is processed in time $O(r)$; and using the
  matroid intersection algorithm of Cunningham~\cite{Cunningham86} in
  \Call{EfficientMatroidCenter}{}, the postprocessing time is $O(r^{3.5})$.
  This gives that total running time is $O((nr + r^{3.5})\ldbe)$.
\end{proof}

We make some remarks.

\begin{remark}
  We do need to know the rank of the matroid (or an upper bound), otherwise we
  cannot control the space usage.  The instances run using a very small guess
  may store a very large number of pivots without the check
  on~\Cref{line:sizecheck}.
\end{remark}

\begin{remark}
  We can decrease the space usage to $O(r^2\log(1/\eps)/\eps)$ points using the
  parallelization ideas of Guha~\cite{Guha09}.  To make the ideas work, we do
  need some properties of matroids.  We give the details
  in~\Cref{sec:handling-guesses}.
\end{remark}

\begin{remark}
  By running $\binom{|E|}{2}$ guesses, \Call{EfficientMatroidCenter}{} can be
  used to get an offline $3$-approximation algorithm for a more general version
  of matroid center, where the cost is computed with respect to a subset $C_E$
  of $E$ and any point in $E$ can be a center.
\end{remark}

\subsection{Extension to Knapsack Center}
\label{sec:extens-knaps-cent}
Recall that in the knapsack center problem, each point $e$ has a non-negative
weight $w(e)$, and the goal is to select a set $C$ of centers that minimizes the
maximum distance between a point and its nearest center subject to the
constraint that $\sum_{c\in C}w(c)\le B$, where $B$ is the \emph{budget}.  We
modify \Cref{alg:1pmc} slightly to give an algorithm for knapsack center using
space $r$ factor smaller than the matroid case, where, in this case, $r$ is the size of a largest
feasible set.  We make sure that all $I_c$ variables are singletons, so the
algorithm stores at most $2r$ points.  Instead of the if condition on
\Cref{line:scaddif}, we replace the point $x$ in $I_c$ by $e$ if $w(x) > w(e)$.
This idea works because replacing a point by a nearby point with a smaller
weight does not affect the feasibility in the knapsack setting (which could
destroy independence in the matroid setting).  Let $C_E$ be the set of pivots at
the end of the stream.  By almost the same argument as in the proof of
\Cref{lem:existence}, we get the following.
\begin{lemma}
\label{lem:existence_knapsack}
  Let $\guess\ge\opt$.  Then there
  exists a feasible set $K\subseteq \bigcup_{c\in C_E}I_c$ such that
  $d(e,K) \le 7\guess$ for any point $e$ and $d(c,K)\le 5\guess$ for any pivot
  $c\in C_E$.
\end{lemma}
For the efficient version, we then use the $3$-approximation algorithm by
Hochbaum and Shmoys~\cite{Hochbaum86}.
\begin{theorem}
  There is an efficient $(17+\eps)$-approximation one-pass algorithm for
  knapsack center that stores at most $4r\log_{(1+\eps/17)} \Delta$ points, where
  $r$ is the size of a largest feasible set.  With a brute force algorithm, one
  can get a $(7+\eps)$-approximation.
  \label{thm:knapsack}
\end{theorem}

\subsection{An Efficient Two Pass Algorithm}
\label{sec:an-efficient-two}

This algorithm is a streaming two-pass simulation of the offline
$3$-approximation algorithm of Chen et al.~\cite{Chen2016} for matroid center.
We describe the algorithm and give the analysis below.

In our one-pass algorithm, i.e.~\Cref{alg:1pmc}, say we are promised that for
any pivot $c$, the optimum center that serves it appears after $c$.  Then it is
enough to try to add $e$ to $I_c$ whenever $d(e,c)\le \guess$; we call this a
modified check.  Let $C_E$ be the set of pivots in the end, then
$(I_c)_{c\in C_E}$ form a partition such that if we pick one point from each
$I_c$ to get set $B$, we can serve each point in $C_E$ using $B$ with cost at
most $\guess$.  With the modified check, for $c,c' \in C_E$ such that
$c\neq c'$, the optimum points $s_c$ and $s_{c'}$ that serve them are also
different because $d(c,c') > 2\guess$.  Now, $s_c \in \spn (I_c)$ due to the
promise that $s_c$ arrived after $c$, and~\Cref{lem:matstruct} gives us the
required independent set $B$.  We then define a partition matroid $\cM_C$ with
partition $(I_c)_{c\in C_E}$ and capacities $1$ and solve the matroid
intersection problem on $\cM_c$ and $\cM$ restricted to $\bigcup_{c\in C_E} I_c$
and get the output $C'$.  Existence of $B$ guarantees that $|C'| = |C_E|$, thus
we are able to serve all points in $C_E$ at a cost of $\guess$.  Since the
points we forget are within distance $2\guess$ of $C_E$, our total cost is at
most $3\guess$ by triangle inequality.  We can get rid of the assumption that
$s_c$ arrives after $c$ by having a second pass through the stream.  We give a
formal description in~\Cref{alg:2pmc}.

\begin{algorithm}
  \begin{algorithmic}[1]
    \Function{MatroidCenter2p}{$\guess$}
    


    \State $C\gets \emptyset$.

    \For{each point $e$ in the stream} \hspace{20pt} \#First pass.

    \If{$d(e,C)\ge 2\guess$}

    \State $C \gets C\cup\{e\}$.

    \State If $\{e\}$ is not a loop, $I_e \gets \{e\}$, else
    $I_e\gets\emptyset$.

    \EndIf

    \EndFor

    \For{each point $e$ in the stream} \hspace{20pt} \#Second pass.

    \If{$\exists c\in C$ such that $d(e,c)\le \guess$ (there can be at most one
      such $c$)}

    \If{$I_c\cup\{e\}$ is independent}

    \State $I_c \gets I_c\cup\{e\}$.

    \EndIf

    \EndIf

    \EndFor

    \State Let $\cM_C=(\bigcup_{c \in C} I_c,\cI_C)$ be a partition matroid with
    partition $\{I_c: c\in C\}$ and capacities $1$.

    \State Let $\cM'$ be the matroid $\cM$ restricted to $\bigcup_{c \in C} I_c$.

    \State $C' \gets $ \Call{matroid-intersection}{$\cM_C$, $\cM'$}

    \If{$|C'| < |C|$}

    \State Return fail with $C$ as certificate.

    \EndIf

    \State Return $C'$.

    \EndFunction
    \caption{Two pass algorithm for matroid center.}
    \label{alg:2pmc}
  \end{algorithmic}
\end{algorithm}

As in the one-pass algorithm, we run $2\log_{1+\eps/3}\Delta$ guesses $\guess$
of $\opt$.  We return the set of centers returned by the instance corresponding
to the smallest guess.  For $\guess\in[\opt,(1+\eps/3)\opt)$, the algorithm will
not abort due to existence of the independent set $B$ (which we argued earlier).
This gives us the following theorem; again we use the matroid intersection
algorithm of Cunningham~\cite{Cunningham86} to get the time bound.

\begin{theorem}
  There is an efficient $(3+\eps)$-approximation two-pass algorithm for matroid
  center that stores at most $2(r^2+r)\log_{(1+\eps/3)} \Delta$ points and has
  total running time $O((nr + r^{3.5})\ldbe)$.
  \label{alg:two-pass-algorithm}
\end{theorem}


%% file: fig_mccost.tex
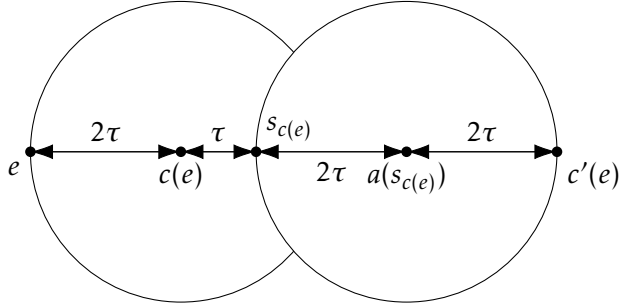
\begin{figure}[!h]
\centering
\begin{tikzpicture}

\draw (0,0) circle (2cm);
\draw[fill=white,draw=none] (3,0) circle (2cm);
\draw (3,0) circle (2cm);

\fill  (0,0) circle (2pt);
\fill  (3,0) circle (2pt);

\fill  (-2,0) circle (2pt);
\fill  (1,0) circle (2pt);
\fill  (5,0) circle (2pt);

\draw[-{Latex[length=3mm, width=2mm]}] (-2,0) -- (0,0);
\draw[-{Latex[length=3mm, width=2mm]}] (0,0) -- (-2,0);

\draw[-{Latex[length=3mm, width=2mm]}] (1,0) -- (0,0);
\draw[-{Latex[length=3mm, width=2mm]}] (0,0) -- (1,0);

\draw[-{Latex[length=3mm, width=2mm]}] (1,0) -- (3,0);
\draw[-{Latex[length=3mm, width=2mm]}] (3,0) -- (1,0);

\draw[-{Latex[length=3mm, width=2mm]}] (5,0) -- (3,0);
\draw[-{Latex[length=3mm, width=2mm]}] (3,0) -- (5,0);

\draw (-2,0) node[anchor=north east] {$e$};

\draw (-1,0) node[anchor=south] {$2\guess$};

\draw (0,0) node[anchor=north] {$c(e)$};

\draw (0.5,0) node[anchor=south] {$\guess$};

\draw (1,0) node[anchor=south west] {$s_{c(e)}$};

\draw (2,0) node[anchor=north] {$2\guess$};

\draw (3,0) node[anchor=north] {$a(s_{c(e)})$};

\draw (4,0) node[anchor=south] {$2\guess$};

\draw (5,0) node[anchor=north west] {$c'(e)$};

\end{tikzpicture}
\caption{To see how to bound the cost of the independent set $B$, let $e$ be any
  point in the stream, $c(e)\in C_E$ be the pivot close to $e$, \,\, $s_{c(e)}$
  be the optimum center that covers $c(e)$, \,\, $a(s_{c(e)})\in C_E$ be the
  pivot close to $s_{c(e)}$, and $c'(e)$ be a point in $B$ that covers
  $a(s_{c(e)})$.}

\label{fig:mccost}
\end{figure}


%% file: outliers_algorithm.tex
\section{Matroid Center with Outliers}
\label{sec:one-pass-algorithm-ouliers}
We first present a simplified analysis of McCutchen and Khuller's
algorithm~\cite{McCutchenK08} for $k$-center with $z$ outliers.  This abstracts
their ideas and sets the stage for the matroid version that we will see later.

\subsection{McCutchen and Khuller's Algorithm}
\label{sec:bett-analys-mccutch}

As usual, we start with a guess $\guess$ for the optimum cost.  The algorithm
maintains a set $C$ of pivots such that $|\bB(c,2\guess)| \ge z+1$ for any
$c\in C$, so the optimum has to serve at least one of these nearby points.
(Recall that $\bB(e,\alpha)=\{x\in E: d(e,x) \le \alpha\}$.)  When a new point
arrives, it is ignored if it is within distance $4\guess$ of $C$.  Otherwise it
is added to the set $F$ of ``free'' points.  As soon as the size of $F$ reaches
$(k - |C| +1)z + 1$, we know for sure that, for a correct guess, the optimum
will have to serve the free points with at most $k - |C|$ clusters, and one of
those clusters will have more than $z$ points by the generalized pigeonhole
principle.  Hence, there \emph{must} exist a free point that has at least $z$
other points within distance $2\guess$ in $F$, because its cluster diameter is at most $2\guess$.  This gives us a new pivot
$c \in F$ with its support points.  We remove those points in $F$ that are
within distance $4\guess$ of $c$ and continue to the next element in the stream.
In the end, we will be left with at most $(k - |C|+1)z$ free points, and they
are served by at most $k - |C|$ optimum centers.  On these remaining free
points, we run an offline $2$-approximation algorithm for $(k - |C|)$-center
with $z$ outliers, e.g., that of Chakrabarty et al.\cite{ChakrabartyGK16}.
\Cref{alg:mkkco} gives a formal description.  We note that we do not need the
sets $A_c$ for $c \in C$ in the algorithm, but we need them in the analysis.

\begin{algorithm}[!h]
  \begin{algorithmic}[1]
    \Function{k-center-z-outliers}{$\guess$}



    \State Pivot-set $C \gets \emptyset$, free-point set $F \gets \emptyset$,
    and $\ell \gets 0$.

    \For{each point $e$ in the stream}

    \If{$d(e,C) > 4\guess$}

    \State $F \gets F \cup \{e\}$.

    \EndIf

    \If{$|F| = (k-\ell+1)z + 1$} \hspace{20pt} \#there is a
    new pivot among the free points;

    \State Let $c \in F$ be such that $|\bB(c,2\guess)\cap F| \ge z+1$
    \hspace{20pt} \#such $c$ will exist for a correct guess.

    \If{such $c$ does not exist}

    \State Abort. \label{line:Fcheck}

    \EndIf

    \State $C \gets C \cup \{c\}$.

    \State $F \gets F \setminus \bB(c,4\guess)$.  \label{line:Freduce}

    \State $A_c \gets \{c\} \cup$ arbitrary subset of
    $\bB(c,2\guess) \setminus \{c\}$ of size $z$.

    \State $\ell \gets \ell + 1$.

    \If{$\ell = k+1$} \hspace{20pt} \# guess is wrong.

    \State Abort.  \label{line:abort1}

    \EndIf

    \EndIf

    \EndFor

    \State $C_F\gets$ $2$-approximation for $(k-\ell)$-center with $z$ outliers
    on $F$ by an efficient offline algorithm.

    \State \Return $C' \gets C \cup C_F$.

    \EndFunction
  \end{algorithmic}
  \caption{McCutchen and Khuller's algorithm~\cite{McCutchenK08} for $k$-center
    with $z$ outliers.}
  \label{alg:mkkco}
\end{algorithm}
Let us bound the space usage first.  The variable $C$ contains at most $k$
pivots, otherwise we abort on~\Cref{line:abort1},
and~\Cref{line:Fcheck,line:Freduce} make sure that the variable $F$ contains at
most $(k+1)z+1$ points.  In total, we store at most $(k+1)z+1$ points at any
moment.

\begin{lemma}
  For $\guess \ge \opt$, \Call{k-center-z-outliers}{$\guess$} stores at most
  $(k+1)z+1$ points, and the cost of $C'$ returned by
  \Call{k-center-z-outliers}{$\guess$} is at most $4\guess$.
  \label{lem:mkkco}
\end{lemma}
\begin{proof}
  Let $C_E$ be the set of pivots and $F_E$ be the set of free points when the
  stream ended, and let $|C_E| = \ell_E$.  We claim that for any $c\neq c'$,
  where $c,c'\in C_E$, $e \in A_c$, and $e'\in A_{c'}$, we have
  $d(e,e') > 2\guess$.  We now prove this claim.  Assume without loss of
  generality that $c$ was made a pivot before $c'$ by the algorithm.  So points
  within distance $4\guess$ of $c$ were removed from $F$.  Any point that
  existed in $F$ after this removal, in particular $e'$, must be farther than
  $4\guess$ from $c$.  This implies that
  \[
    4\guess < d(c,e') \le d(c,e) + d(e,e'), \;\;\;\text{ and }\;\;\;
    d(e,e') > 2\guess\,,
  \]
  because $d(c,e)\le 2\guess$.  Now, we know that for $c \in C_E$, there exists
  $x_c \in A_c$ that has to be served by an optimum center, say $s_c$, because
  $|A_c| > z$, so not all of the points in $A_c$ can be outliers.  By the
  earlier claim, for $c\neq c'$, we have $d(x_c,x_{c'}) > 2\guess$ implying that
  $s_c\neq s_{c'}$ and $\ell_E \le k$.  Also note that none of these optimum
  centers can serve a point in $F_E$, because by triangle inequality
  \[
    d(s_c,F_E) \ge d(c,F_E) - d(c,x_c) - d(x_c,s_c) > 4\guess - 2\guess - \guess
    = \guess
  \]
  for $c\in C_E$.  This shows that all but $z$ points in $F_E$ have to be served
  by at most $k - \ell_E$ optimum centers with cost at most $\guess$.  For each
  of these optimum centers, there exists a free point in $F_E$ within distance
  $\guess$.  So there exists a set $B_F$ of $k-\ell_E$ points in $F_E$, such
  that $B_F$ covers all but at most $z$ points of $F_E$ with cost $2\guess$.  So
  a $2$-approximation algorithm recovers $k - \ell_E$ centers with cost at most
  $4\guess$.  Observing that we only forget points in the stream that are within
  distance $4\guess$ of some pivot in $C_E$ finishes the proof.
\end{proof}

By running \Call{k-center-z-outliers}{$\guess$} for at most
$O(\log(1/\eps)/\eps)$ geometrically-increasing active guesses, we get the
$(4+\eps)$-approximation algorithm for $k$-center with $z$ outliers.  This
analysis is based on that of McCutchen and Khuller~\cite{McCutchenK08}.


\subsection{Matroid Center with Outliers}
\label{sec:matroid-center-with}

It is now possible to naturally combine the ideas used for matroid center and
those used for $k$-center with $z$ outliers to develop an algorithm for matroid
center with $z$ outliers.

Whenever the free-point set becomes large enough, we create a pivot $c$ and an
independent set $I_c$ to which we try to add all free points within distance
$4\guess$ of $c$.  We do the same for a new point $e$ in the stream that is
within distance $4\guess$ of some pivot $c \in C$, i.e., we try to add it to
$I_c$ keeping $I_c$ independent in the matroid.  Otherwise $d(e,C) > 4\guess$,
so we make it a free point.  The structural property of matroids that we proved
as~\Cref{lem:matstruct} then enables us to show that $\bigcup_{c\in C} I_c$ and
the set of free points make a good summary of the stream.  See~\Cref{alg:1pmcwo}
for a formal description.  Here, we note that we do not need the sets $A_c$ for
$c \in C$ in the algorithm if flag is set to ``brute force'', but we need them
in the analysis in any case.

\begin{algorithm}[!h]
  \begin{algorithmic}[1]
    \Function{matroid-center-z-outliers}{$\guess$, flag}



    \State Pivot-set $C \gets \emptyset$, free-point set $F \gets \emptyset$,
    and $\ell \gets 0$.

    \For{each point $e$ in the stream}

    \If{$\exists c\in C$ such that $d(e,c)\le 4\guess$}
      
    \If{$I_c\cup\{e\}$ is independent}

    \State $I_c \gets I_c\cup\{e\}$.  \label{line:scadd1}

    \EndIf
          
    \Else

    \State $F \gets F\cup\{e\}$.

    \EndIf

    \If{$|F| = (r-\ell+1)z + 1$}

    \State Let $c \in F$ be such that $|\bB(c,2\guess)\cap F| \ge z+1$ (if not,
    we guessed wrong, so abort).

    \State $C \gets C \cup \{c\}$.

    \State $A_c\gets \{c\}$ and if $\{c\}$ is not a loop, $I_c \gets \{c\}$,
    else $I_c\gets\emptyset$.

    \State $\ell \gets \ell + 1$ (if $\ell$ becomes $r+1$ here, we guessed
    wrong, so abort).

    \For{each $x \in F\cap \bB(c,4\guess)$}

    \State $F \gets F\setminus\{x\}$.
        
    \If{$I_c\cup\{x\}$ is independent}

    \State $I_c \gets I_c\cup\{x\}$.  \label{line:scadd2}

    \EndIf

    \If{$|A_c| \le z$}

    \State $A_c \gets A_c\cup\{x\}$.

    \EndIf

    \EndFor

    \EndIf

    \EndFor

    \If{flag = ``brute force''}

    \State Find an independent set $C'_B$ in $F \cup \bigcup_{c\in C}I_c$ by
    brute force such that cost of $C'_B$ is at most $11\guess$ \\
    \hspace{50pt}with respect to $C$ and at most $9\guess$ with respect to all
    but at most $z$ points of $F$.

    \State If such $C'_B$ does not exist, abort, else return $C'_B$.

    \EndIf

    \If{flag = ``efficient''}

    \State Run the offline $3$-approximation algorithm by Harris et
    al.~\cite{Harris17} for matroid center with \\
    \hspace{50pt}$z$ outliers to get an independent set $C'$ of centers in
    $F \cup \bigcup_{c\in C}({A_c \cup I_c})$ such that cost of $C'$  \\
    \hspace{50pt}is at most $47\guess$ with respect to $C$ and at most $45\guess$
    with respect to all but $z$ points of $F$.

    \State If such $C'$ does not exist, abort, else return $C'$.

    \EndIf

    \EndFunction
    \caption{One-pass algorithm for matroid center with outliers.}
    \label{alg:1pmcwo}
  \end{algorithmic}
\end{algorithm}
Let $C_E$ be the set of pivots and $F_E$ be the set of free points when the
stream ended, and let $\ell_E = |C_E|$.



\begin{lemma}[Small summary for matroid center with outliers] For
  $\guess \ge \opt$, \Cref{alg:1pmcwo} stores at most $O(r^2 +rz)$ points, and
  there exists an independent set $B\subseteq F_E \cup \bigcup_{c\in C_E}I_c$
  such that cost of $B$ is at most $15\guess$; also $d(c,B)\le 11\guess$ for any
  pivot $c\in C_E$, and $B$ covers all but at most $z$ points of $F_E$ with cost
  at most $9\guess$.
  \label{lem:existence_wo}
\end{lemma}
\begin{proof}
  Let $I^*$ be an optimum independent set of centers. By the same argument as in
  the proof of~\Cref{lem:mkkco}, the following claim is true.  For any
  $c\neq c'$, where $c,c' \in C_E$, $e \in A_c$, and $e'\in A_{c'}$, we have
  $d(e,e') > 2\guess$.  Now, we know that for $c \in C_E$, there exists
  $x_c \in A_c$ that has to be served by an optimum center, say $s_c$, because
  $|A_c| > z$.  By the earlier claim, for $c\neq c'$, we have
  $d(x_c,x_{c'}) > 2\guess$ implying that $s_c\neq s_{c'}$ and $\ell_E \le r$.
  Let $I^*_{C_E}=\{s_c : c \in C_E\}$ be the set of optimum centers that serve
  some $x_c\in A_c$ for $c\in C_E$.  None of the optimum centers in $I^*_{C_E}$
  can serve a point in $F_E$, because $d(s_c,F_E) > \guess$ for $c\in C_E$.
  This shows that all but $z$ points in $F_E$ have to be served by at most
  $r - \ell_E$ optimum centers with cost at most $\guess$.  Since $|I_c|\le r$
  for any $c$ in the variable $C$, size of $\bigcup_{c\in C}I_c$ is always
  bounded by $r^2$.  Also, the check on the size of $F$ ensures that
  $|F| \le (r + 1)z + 1$, so total number of points stored is at most
  $O(r^2 + rz)$ at any moment.

  When we first process a new point $e$ in the stream, we either try to add it
  to some $I_c$ or to $F$.  If $e$ is never removed from $F$, then $e\in F_E$,
  otherwise, we try to add it to some $I_c$.  The same argument applies to any
  $x \in I^*$, so if $x \notin F_E$, then we did try to add it to some $I_c$.
  For an $x \in I^*\setminus F_E$, let $a(x)\in C_E$ denote the pivot whose
  independent set $I_{a(x)}$ we tried to add $x$ to.

  By~\Cref{lem:matstruct}, for
  $\mathcal{A} := \{I_{a(x)} : x \in I^*\setminus F_E\} \cup \{\{x\}: x \in
  I^*\cap F_E\}$, there exists an independent set $B$ such that
  $|I \cap B| \ge 1$ for all $I\in\mathcal{A}$.  Since for $x \in I^*\cap F_E$
  the singleton $\{x\} \in \mathcal{A}$, the set $B$ must contain $\{x\}$.  For
  a free point $e$ served by an optimum center $s$ such that we tried to add $s$
  to some $I_c$, we have that
  $d(e,B) \le d(e,s) + d (s,c) + d(c,B)\le \guess + 4\guess +4\guess = 9\guess$,
  which means that $B$ serves all but $z$ points of $F_E$ with cost at most
  $9\guess$.  Now, we claim that for any point $e$ in the stream,
  $d(e,B) \le 15\guess$.  We just saw that if $e\in F_E$ is served by an optimum
  center, then $d(e,B) \le 9\guess$, so assume that $e\notin F_E$, that means
  there is a $c\in C_E$ such that $d(e,c) \le 4\guess$; denote this $c$ by
  $c(e)$.  See~\Cref{fig:mccost_wo}.  Let $s_{c(e)}$ be the optimum center that
  serves an $x_{c(e)} \in A_{c(e)}$ (recall that such a point exists because
  $|A_{c(e)}| > z$).  So $d(c(e),s_{c(e)}) \le 3\guess$, and $a(s_{c(e)})\in C_E$
  was the pivot such that $d(s_{c(e)},a(s_{c(e)})) \le 4\guess$.  Let $c'(e)$ be
  an arbitrary point in $I_{a(s_{c(e)})}\cap B$, whose existence is guaranteed
  by the property of $B$.  We have $d(a(s_{c(e)}),c'(e)) \le 4\guess$, because
  $c'(e)\in I_{a(s_{c(e)})}$.  Then by triangle inequality,
  \begin{align*}
    d(e,B) &\le d(e,c'(e))\\
           &\le d(e,c(e)) + d(c(e),x_{c(e)}) + d(x_{c(e)},s_{c(e)}) +
             d(s_{c(e)},a(s_{c(e)})) + d(a(s_{c(e)}),c'(e)) \\
           &\le 4\guess + 2\guess + \guess + 4\guess + 4\guess = 15\guess\,,
  \end{align*}
  hence, cost of $B$ is at most $15\guess$.

  For any $c\in C_E$, we can bound $d(c,B)$ in a similar way.  Let $s_c$ be the
  optimum center that serves an $x_c\in A_c$.  Define $a(s_c)$ to be the pivot
  such that $d(s_c,a(s_c)) \le 4\guess$.  Also, let $c'$ be the point in $B$
  such that $d(a_{s_c},c') \le 4\guess$.  This gives that
  $d(c,B) \le d(c,c') \le 11\guess$.  We already established that $B$ covers all
  but at most $z$ points of $F_E$ with cost at most $9\guess$.  The proof is now
  complete.
\end{proof}

\input{fig_mccost_wo}

\begin{theorem}
  There is an efficient $(51+\eps)$-approximation one-pass algorithm for matroid
  center with $z$ outliers that stores at most $O((r^2+rz)\log \Delta/\eps)$ points.  With
  a brute force algorithm, one can get a $(15+\eps)$-approximation.
  \label{thm:mczoalgo}
\end{theorem}
\begin{proof}
  We run $O(\log \Delta/\eps)$ parallel copies of
  \Call{matroid-center-z-outliers}{$\guess$, flag} and return the output of the
  copy for the smallest unaborted guess.  We claim that the copy corresponding
  to guess $\guess'\in [\opt,(1+\eps/50)\opt)$, call it $\inst(\guess')$, will
  not abort.  Denote by $C_E$, $F_E$, and $(I_c)_{c\in C_E}$ contents of the
  corresponding variables in $\inst(\guess')$ at the end of the stream (we will
  not abort mid-stream because $\guess'\ge\opt$).

  By \Cref{lem:existence_wo}, $F_E \cup \bigcup_{c\in C_E}I_c$ contains a
  solution that has cost $11\guess'$ with respect $C_E$ and $9\guess'$ with
  respect to all but at most $z$ of $F_E$.  These checks can be performed by the
  brute force algorithm.  Since any instance for guess $\guess$ forgets only
  those points within distance $4\guess$ of its pivots, the brute force
  algorithm outputs a $(15+\eps)$-approximation.

  By~\Cref{lem:existence_wo}, there exists a solution of cost at most
  $15\guess'$, and the efficient $3$-approximation algorithm for matroid center
  with $z$ outliers will return a solution $C'$ with cost at most $45\guess'$.
  Note that $C'$ has to cover at least one point from $A_c$ for each $c\in C_E$,
  hence $d(c,C')\le 47\guess'$.  Since we forget points only within distance
  $4\guess'$ of $C_E$, we get the desired approximation ratio.
\end{proof}


%% file: fig_mccost_wo.tex
\begin{figure}[!h]
\centering
\begin{tikzpicture}

\draw (0,0) circle (3cm);
\draw[fill=white,draw=none] (5,0) circle (3cm);
\draw (5,0) circle (3cm);

\fill  (0,0) circle (2pt);
\fill  (5,0) circle (2pt);

\fill  (-3,0) circle (2pt);
\fill  (1,0) circle (2pt);
\fill  (2,0) circle (2pt);
\fill  (5,0) circle (2pt);
\fill  (8,0) circle (2pt);

\draw[-{Latex[length=3mm, width=2mm]}] (-3,0) -- (0,0);
\draw[-{Latex[length=3mm, width=2mm]}] (0,0) -- (-3,0);

\draw[-{Latex[length=3mm, width=2mm]}] (1,0) -- (0,0);
\draw[-{Latex[length=3mm, width=2mm]}] (0,0) -- (1,0);

\draw[-{Latex[length=3mm, width=2mm]}] (1,0) -- (2,0);
\draw[-{Latex[length=3mm, width=2mm]}] (2,0) -- (1,0);

\draw[-{Latex[length=3mm, width=2mm]}] (5,0) -- (2,0);
\draw[-{Latex[length=3mm, width=2mm]}] (2,0) -- (5,0);

\draw[-{Latex[length=3mm, width=2mm]}] (5,0) -- (8,0);
\draw[-{Latex[length=3mm, width=2mm]}] (8,0) -- (5,0);

\draw (-3,0) node[anchor=north east] {$e$};

\draw (-1.5,0) node[anchor=south] {$4\guess$};

\draw (0,0) node[anchor=north] {$c(e)$};

\draw (0.5,0) node[anchor=south] {$2\guess$};

\draw (1,0) node[anchor=north] {$x_{c(e)}$};

\draw (1.5,0) node[anchor=south] {$\guess$};

\draw (2,0) node[anchor=south west] {$s_{c(e)}$};

\draw (3.5,0) node[anchor=south] {$4\guess$};

\draw (5,0) node[anchor=north] {$a(s_{c(e)})$};

\draw (6.5,0) node[anchor=south] {$4\guess$};

\draw (8,0) node[anchor=north west] {$c'(e)$};

\end{tikzpicture}
\caption{To see how to bound the cost of the independent set $B$, let $e$ be any
  point in the stream, $c(e)\in C_E$ be the pivot close to $e$, \,\, $x_{c(e)}$
  be the point in the support $A_{c(e)}$ of $c(e)$ that an optimum center
  serves, $s_{c(e)}$ be the optimum center that serves $x_{c(e)}$,\,\,
  $a(s_{c(e)})\in C_E$ be the pivot close to $s_{c(e)}$, and $c'(e)$ be a point
  in $B$ that covers $a(s_{c(e)})$.}
\label{fig:mccost_wo}
\end{figure}
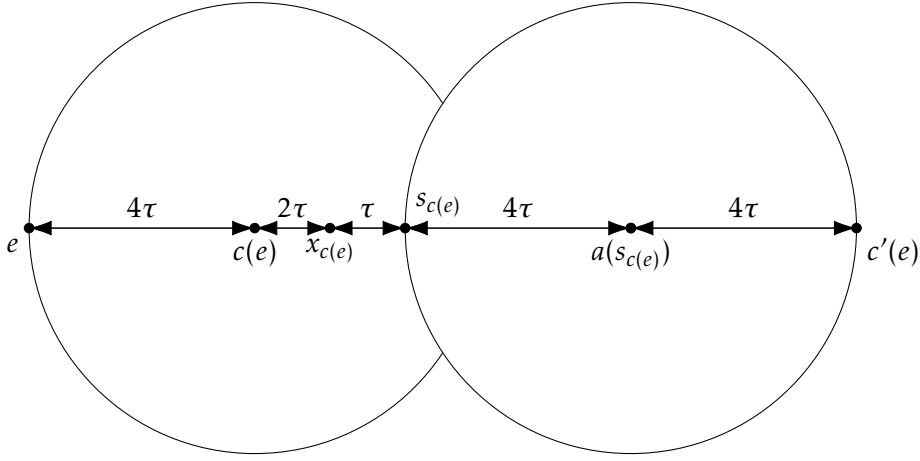


%% file: guessing.tex
\section{Handling the Guesses}
\label{sec:handling-guesses}

We extend the ideas of Guha~\cite{Guha09} and McCutchen and
Khuller~\cite{McCutchenK08} to run $O(\log(1/\eps)/\eps)$ active guesses.
Although, to make this idea work for matroids, we do need a property of matroids
(see~\Cref{lem:nspanso}).  The way to do this is to start with a lower bound $R$
on the optimum and spawn instances, which we call \emph{original} instances,
$\inst(\guess)$ for guesses
$\guess = R, R(1+\eps), \ldots, R(1+\eps)^\beta = R\alpha/\eps$, for some
$\alpha$ that depends on the basic algorithm that we use, e.g., for matroid
center, we will use $\alpha=2+\eps$.  When a guess $\guess'$ fails, we replace
an instance $\inst = \inst(\guess)$ for $\guess \le \guess'$ with a new
instance, which we call its \emph{child} instance,
$\inst_N = \inst(\guess(1+\eps)^\beta)$.  In the new instance $\inst_N$, we
treat the summary that we maintained for $\inst(\guess)$ as the initial stream.
Since the new guess in $\inst_N$ is about $1/\eps$ times larger than the old
guess in $\inst$, the distance between a point that we forgot and the summary
stored by $\inst$ is about $\eps$ times the new guess.  Therefore, the cost
analysis does not get much affected for a correct guess.  If we forgot an
optimum center, a nearby point in the summary can act as its replacement.  This
statement is obvious for a uniform matroid, because all points are treated the
same way within the matroid, but it is not true for general matroids; in fact,
as exhibited by our lower bound, it is not true even for partition matroids.  So
with each point in the summary, we pass to the new instance an independent set
$I_o$.  The following simple lemma shows that if an optimum center $x$ is in the
span of $I_o$, and if we construct $I_c$ for a new pivot $c$ such that
$I_o\subseteq \spn(I_c)$, then $I_c$ also spans the optimum center.

\begin{lemma}
  Let $I$ and $J$ be independent sets of a matroid such that
  $J\subseteq\spn(I)$.  If $e \in \spn(J)$, then $e\in\spn(I)$.
  \label{lem:nspanso}
\end{lemma}
\begin{proof}
  Let $\rank(I) = q$.  Towards a contradiction, let $\rank(I\cup \{e\}) = q+1$.
  Since $J\subseteq\spn(I)$, $\rank(I\cup J) = q$.  Now, $e \in \spn(J)$, so
  $\rank(I\cup J \cup \{e\}) = q$, i.e., $\rank(I\cup \{e\}) \le q$, which gives
  us the desired contradiction. 
\end{proof}

\subsection{A Smaller Space Algorithm for Matroid Center}
We modify the function \Call{MatroidCenter}{$\guess$,flag} from earlier to
accept a starting stream and an independent set for each point in the starting
stream: \Call{MatroidCenter}{$\guess$, $C_o$, $(J_{c_o})_{c_o\in C_o}$,flag}.
Before processing any new points in the stream we process the points in $C_o$ as
follows.  When processing a $c_o\in C_o$, if $d(c_o,C) \le 2\guess$, try to add
points in $J_{c_o}$ to $I_c$.  Otherwise create a new pivot $c$ in $C$ and
initialize $I_c = J_{c_o}$.  Once $C_o$ is processed, we continue with the
stream and work exactly as in \Call{MatroidCenter}{$\guess$}.  We give complete
pseudocode in~\Cref{alg:1pmcss}.

\begin{algorithm}[!h]
  \begin{algorithmic}[1]

    \State Let $R$ be the minimum distance for some two points in the first
    $r+1$ points in the stream.

    \For{$\guess\in \{R, R(1+\eps),\ldots, R(1+\eps)^\beta = (2+\eps)R/\eps\}$
      in parallel}

    \State \Call{MatroidCenter}{$\guess$, $\emptyset$, $\emptyset$}.

    \If{an instance with guess $\guess$ is aborted}

    \For{all active $\inst(\guess')$ with guess $\guess' \le \guess$, current
      pivots $C_o$, and independent sets $(J_{c_o})_{c_o\in C_o}$}

    \State Replace it with the child instance
    \Call{MatroidCenter}{$\guess'(1+\eps)^\beta$, $C_o$,
      $(J_{c_o})_{c_o\in C_o}$, flag}.

    \EndFor

    \EndIf

    \EndFor

    \State Return the set $C'$ of centers returned by the active instance with
    the smallest guess. \label{line:ss:return}

    \State \mbox{}
    
    \Function{MatroidCenter}{$\guess$, $C_o$, $(J_{c_o})_{c_o\in C_o}$, flag}
    


    \State $C\gets \emptyset$.

    \For{each point $c_o$ in $C_o$}  \label{line:ss:ngs}

    \If{$\exists c\in C$ such that $d(c_o,c)\le2 \guess$ (pick arbitrary such
      $c$ if there are several)}

    \For{$e_o \in J_{c_o}$}

    \If{$I_c\cup\{e_o\}$ is independent}

    \State $I_c \gets I_c\cup\{e_o\}$.  

    \EndIf

    \EndFor

    \Else

    \State $C \gets C\cup\{c_o\}$.

    \State $I_{c_o} \gets J_{c_o}$. \label{line:ss:nge}

    \EndIf

    \EndFor

    \State \# Processing of the old pivots finished, continue with the actual
    stream.

    \For{each point $e$ in the stream}

    \If{there is a pivot $c\in C$ such that $d(e,c)\le2 \guess$ (pick arbitrary
      such $c$)}

    \If{$I_c\cup\{e\}$ is independent}

    \State $I_c \gets I_c\cup\{e\}$.  \label{line:ss:scadd}

    \EndIf

    \ElsIf{$|C| = r$} \label{line:ss:sizecheck} \Comment{We cannot have more
      pivots than the rank.}

    \State Abort. \Comment{Because $C\cup \{e\}$ acts as a certificate that the
      guess is incorrect.}

    \Else

    \State $C \gets C\cup\{e\}$.  \Comment{Make $e$ a pivot.}

    \State If $\{e\}$ is not a loop, $I_e \gets \{e\}$, else
    $I_e\gets\emptyset$.

    \EndIf

    \EndFor

    \If{flag = ``brute force''}

    \State Find an independent set $C'_B$ in $\bigcup_{c\in C}I_c$ such that
    $d(c,C'_B)\le (5+2\eps)\guess$ for $c \in C$.

    \State If such $C'_B$ does not exist, then abort, else return
    $C'_B$. \label{line:ss:bf}

    \EndIf

    \State \Return \Call{EfficientMatroidCenter}{$(5+2\eps)\guess$, $C$,
      $(I_c)_{c\in C}$, $\cM$}.

    \EndFunction
    \caption{One pass algorithm for matroid center with smaller space.}
    \label{alg:1pmcss}
  \end{algorithmic}
\end{algorithm}
For an instance $\inst(\guess)$ let $C_o(\guess)$ be the initial summary and
$\cJ(\guess)$ be the collection of independent sets that we passed to it, and
let $E(\guess)$ be the part of the actual stream that it processed.  Also, let
$\inst(\guess_o)$ be the instance for $\guess_o = \eps\guess/(2+\eps)$ from
which $\inst(\guess)$ was spawned.
\begin{lemma}
  Let $e$ be a point that arrived before the substream $E(\guess)$. Then $e$ has
  a nearby \emph{representative} $\rho_e \in C_o(\guess)$ such that
  $d(e,\rho_e) \le \eps\guess$ and also the independent set $J_{\rho_e}$
  corresponding to $\rho_e$ spans $e$.
  \label{lem:club}
\end{lemma}
\begin{proof}
  We prove this claim by induction on the number of ancestors.  For an original
  instance, the claim holds trivially, because no point arrived before.
  Otherwise, there are two cases: either $e\in E(\guess_o)$ or $e$ arrived
  before $E(\guess_o)$.  If $e\in E(\guess_o)$, then by the logic of the
  algorithm, there exists $a(e) \in C_o(\guess)$ such that
  $d(e,a(e))\le 2\guess_o = 2 \eps\guess/(2+\eps)\le \eps \guess$, and also we
  tried to add $e$ to $I_{a(e)}$ (that became $J_{a(e)}$ for the next instance
  $\inst(\guess)$).  Otherwise, by induction hypothesis, there is a point
  $e'\in C_o(\guess_o)$ such that $d(e,e') \le \eps\guess_o$ and $J_{e'}$ spans
  $e$.  Now, let $\rho_{e'}\in C_o(\guess)$ be such that
  $d(e',\rho_{e'})) \le 2 \guess_o$ (such $\rho_{e'}$ must exist by logic of the
  algorithm).  Using triangle inequality and the above inequality that
  $d(e,e') \le \eps\guess_o$, we get
  \[
    d(e,\rho_{e'}) \le d(e,e') + d(e',\rho_{e'}) \le \eps\guess_o + 2\guess_o =
    (2+\eps)\guess_o = (2+\eps)\frac{\eps\guess}{(2+\eps)} = \eps\guess\,.
  \]
  Moreover, in the instance $\inst(\guess_o)$, we tried to add all points in
  $J_{e'}$ to $I_{\rho_{e'}}$, so by~\Cref{lem:nspanso},
  $e\in\spn(I_{\rho_{e'}})$ (see that $I_{\rho_{e'}}$ became $J_{\rho_{e'}}$ for
  the next instance $\inst(\guess)$), which proves the claim.
\end{proof}
\begin{theorem}\label{thm:guessmc}
  There is an efficient $((17+7\eps)(1+\eps))$-approximation one-pass algorithm
  for matroid center that stores at most $O(r^2\lobebe)$
  points and has
  total running time $O((nr + r^{3.5})\lobebe + r^2\ldbe)$.  With a brute force algorithm, one can get a
  $((7+3\eps)(1+\eps))$-approximation in total running time $O((nr + r^{2r+2})\lobebe + r^2\ldbe)$.
\end{theorem}
\begin{proof}
  Space usage is easy to analyze.  At any time, we have at most
  $O(\log_{1+\eps}(1/\eps)) = O(\lobebe)$ active instances and each instance stores at most
  $O(r^2)$ points.

  Running time analysis for both, brute force and efficient algorithm, is the
  same as that in the proof of \Cref{thm:main}, except the additional $r^2\ldbe$
  term appears because it takes time $O(r^2)$ to initialize a new guess starting
  \Cref{line:ss:ngs} till \Cref{line:ss:nge}, and there can be at most
  $O(\ldbe)$ possible guesses.

  Consider the instance $\inst(\guess')$ for which we returned
  on~\Cref{line:ss:return} in~\Cref{alg:1pmcss}, and suppose the outputs were
  $C'$ or $C'_B$ (depending on ``flag'').  We note that some active copy
  \emph{will} return, because $\guess$ cannot keep on increasing indefinitely.
  E.g., consider $\guess$ larger than the maximum distance between any two
  points.  Let $C_E$ be the contents of the variable $C$ in $\inst(\guess')$ at
  the end of the stream.  Then we know that costs of $C'_B$ and $C'$ are at most
  $(5+2\eps)\guess'$ and $(15+6\eps)\guess'$ with respect to $C_E$ due to the
  check that we do on~\Cref{line:ss:bf} and by~\Cref{thm:chen} for
  \Call{EfficientMatroidCenter}{}.  By \Cref{lem:club}, any point that arrived
  before $E(\guess')$ is within distance $\eps\guess'$ of $C_o(\guess')$, and
  each point in $C_o(\guess')$ is within distance $2\guess'$ of $C_E$, which
  shows that costs of $C'_B$ and $C'$ are at most $(7+3\eps)\guess'$ and
  $(17+7\eps)\guess'$ with respect to the whole stream (by triangle inequality).
  Next, we show that $\guess'\le (1+\eps)\opt$, and that will finish the proof.

  Consider the guess $\guess \in (\opt,(1+\eps)\opt]$.  If $\guess$ was never
  active, that means $\guess'\le \opt$, and we are done.  Otherwise, $\guess$
  was active, and we will prove that it was not aborted.  Since $\guess\le\opt$,
  we will not abort mid-stream in $\inst(\guess)$, so let $C_E$ be the set of
  pivots at the end of the stream in $\inst(\guess)$.  We will show that there
  is an independent set $B$ such that cost of $B$ with respect to $C_E$ is at
  most $(5+2\eps)\guess$.  By \Cref{line:ss:bf} and by~\Cref{thm:chen} for
  \Call{EfficientMatroidCenter}{}, this would imply that $\inst(\guess)$ cannot
  abort.

  From here on, the proof follows that of~\Cref{lem:existence}.  Let $c\in C_E$.
  Denote by $s_c$ the optimum center that serves it, so $d(c,s_c) \le \guess$.
  If $s_c \in E(\guess)$, then $s_c \in\spn(I_{c'})$ for some $c'\in C_E$ and
  $d(s_c,c') \le 2\guess$.  Otherwise, $s_c$ arrived before $E(\guess)$.  Let
  $\rho_{s_c}$ be the representative of $s_c$ whose existence is guaranteed by
  \Cref{lem:club}, so $d(s_c,\rho_{s_c})\le\eps\guess$.  Then let $c' \in C_E$
  be such that $d(\rho_{s_c},c')\le 2\guess$ and $J_{\rho_{s_c}}$ is spanned by
  $I_{c'}$.  Thus, by triangle inequality
  \begin{equation}
    \label{eq:3}
    d(s_c,c')\le d(s_c,\rho_{s_c}) + d(\rho_{s_c},c') \le \eps\guess + 2\guess =
    (2+\eps)\guess\,,
  \end{equation}
  and by~\Cref{lem:nspanso}, $s_c$ is spanned by $I_{c'}$.  Denote by
  $\mathcal{A}$ the collection of such $I_{c'}$'s.  Now,
  by~\Cref{lem:matstruct}, there exists an independent set $B$ such that
  $|I \cap B| \ge 1$ for all $I\in\mathcal{A}$.  Pick $c_p$ from $I_{c'}\cap B$.
  Either $c_p\in E(\guess)$ or it arrived before.  In any case, again using
  \Cref{lem:club}, we have $d(c_p,c')\le(2+\eps)\guess$ (we use this below),
  and
  \begin{itemize}
    \item $d(c,s_c) \le \guess$, because $s_c$ is the optimum center that covers
    $c$,
    \item $d(s_c,c') \le (2+\eps)\guess$, by~Inequality~\eqref{eq:3}, and
    \item $d(c',c_p)\le (2+\eps)\guess$.
  \end{itemize}
  Thus, by triangle inequality, $d(c,B) \le (5+2\eps)\guess$.  So
  $\inst(\guess)$ will not abort.  This finishes the proof.
\end{proof}

This technique also gives a better two-pass algorithm.
\begin{theorem}
  There is an efficient $(3+\eps)$-approximation two-pass algorithm for matroid
  center that stores at most $O(r^2\lobebe)$ points and has
  total running time $O((nr + r^{3.5})\lobebe + r^2\ldbe)$.
  \label{alg:two-pass-algorithm-better}
\end{theorem}

Reducing the space usage for matroid center with $z$ outliers can be done by
naturally combining the techniques above and those
in~\Cref{sec:matroid-center-with}.  We define a similar overloading
\Call{matroid-center-z-outliers}{$\guess$, $C_o$, $(J_{c_o})_{c_o\in C_o}$,
  $F_o$, flag}, where $F_o$ contains the set of free points in $\inst(\guess_o)$
when it aborted and this function was called with the updated guess $\guess$.
We skip the details and state the following theorem without proof.
\begin{theorem}
  There is an efficient $(51+\eps)$-approximation one-pass algorithm for matroid
  center with $z$ outliers that stores at most $O((r^2+rz)\lobebe)$ points.  With
  a brute force algorithm, one can get a $(15+\eps)$-approximation.
  \label{thm:mczoguess}
\end{theorem}

\subsection*{Extension to Knapsack Center}
\label{sec:extens-knaps-cent-1}

In \Cref{sec:extens-knaps-cent}, we saw how to modify \Cref{alg:1pmc} to get an
algorithm for knapsack center that stores at most $2r$ points, where $r$ is the
size of a largest feasible set.  Using the same idea, algorithms for two-pass
matroid center, matroid center with outliers, and smaller space matroid center,
which are \Cref{alg:2pmc,alg:1pmcwo,alg:effchen}, can be extended to the
knapsack center without losing the approximation ratio and with a space $r$ factor smaller than the matroid case.  For the outlier version of knapsack center, to get an efficient
algorithm, we use the $3$-approximation algorithm by Chakrabarty and
Negahbani~\cite{chakrabartyN18}.  So we get the following theorems, where $r$ is
the size of a largest feasible set.

\begin{theorem}\label{thm:guesskc}
  There is an efficient $(17+\eps)$-approximation one-pass algorithm for
  knapsack center that stores at most $O(r\lobebe)$ points.  With a brute force
  algorithm, one can get a $(7+\eps)$-approximation.
\end{theorem}

\begin{theorem}
  There is an efficient $(51+\eps)$-approximation one-pass algorithm for
  knapsack center with $z$ outliers that stores at most $O(rz\lobebe)$ points.  With a
  brute force algorithm, one can get a $(15+\eps)$-approximation.
  \label{thm:kczoguess}
\end{theorem}


%% file: ack.tex
\subparagraph*{Acknowledgements.}  We thank Ashish Chiplunkar for his
contributions, Maryam Negahbani for discussions, and anonymous reviewers for helpful comments.


%% file: offline.tex
\section{An Implementation of Efficient Matroid Center}
\label{sec:effic-appr-algor}

We now give an implementation of \Call{EfficientMatroidCenter}{}.  The input
consists of $\offguess$, $C_E$, $X$, such that $C_E\subseteq X$, and the
underlying matroid $\cM$ defined over $X$.  Furthermore, the promise is that
there is an independent set $B\subseteq X$ such that for each $c\in C_E$, we
have $d(c,B)\le \offguess$.  Our implementation is based on the algorithm of
Chen et al.~\cite{Chen2016} for matroid center.  We show that it outputs a set
$C'$ such that, assuming the promise, $d(c,C')\le 3\offguess$ for $c\in C_E$.
\begin{algorithm}
  \begin{algorithmic}[1]
    \Function{EfficientMatroidCenter}{$\offguess$, $C_E$, $X$,
      $\cM$}

    \State Initialize: $C\gets \emptyset$.

    \While{there is an unmarked point $e$ in $C_E$}

    \State $C \gets C \cup \{e\}$, $B_e \gets \bB(e,\offguess)\cap X$, and mark
    all points in $\bB(e,2\offguess)\cap C_E$. \label{line:ol:mark}

    \EndWhile

    \State Let $\cM_C=(\cup_{c \in C} B_c,\cI_C)$ be a partition matroid with
    partition $\{B_c: c\in C\}$ and capacities $1$.

    \State Let $\cM'$ be the matroid $\cM$ restricted to $\cup_{c \in C} B_c$.

    \State $C' \gets $ \Call{matroid-intersection}{$\cM_C$, $\cM'$}

    \If{$|C'| < |C|$}

    \State Return fail.

    \EndIf

    \State Return $C'$.

    \EndFunction
    \caption{Efficient algorithm for matroid center based on the algorithm
      by~\cite{Chen2016}.}
    \label{alg:effchen}
  \end{algorithmic}
\end{algorithm}

\begin{theorem}
  If \Call{EfficientMatroidCenter}{} does not fail, then it outputs a set $C'$
  such that $d(c,C')\le 3\offguess$ for each $c\in C_E$.  If the input
  promise holds, then \Call{EfficientMatroidCenter}{} does not fail.
  \label{thm:chencopy}
\end{theorem}
\begin{proof}
  In this proof, we refer by $C$ the contents of the variable $C$ after the
  while loop ended, and let $c_E$ be any arbitrary point in $C_E$.  Define the
  function $\Mark : C_E \rightarrow C$ such that $\Mark(c_E) \in C$ is the
  ``marker'' of $c_E$, i.e., we marked $c_E$ when processing $\Mark(c_E)$.  In
  the end, all $c_E$'s are marked, so $\Mark$ is a valid function.  By the logic
  on~\Cref{line:ol:mark}, we have that
  \begin{equation}
    \label{eq:1}
    d(c_E,\Mark(c_E))\le 2\offguess\,.    
  \end{equation}

  Let \Call{EfficientMatroidCenter}{} does not fail, then $|C'| \ge |C|$ and
  $C'$ satisfies the partition matroid constraint of $\cM_C$.  By definition of
  $\cM_C$, $\rank(\cM_C) = |C|$, hence $|C'|\le |C|$, which implies that
  $|C'|=|C|$. Therefore, for each $c\in C$, the set $C'$ must contain exactly
  one element in $\bB(c,\offguess)$ and $d(c,C')\le \offguess$, in particular,
  $d(\Mark(c_E),C')\le \offguess$.  This, triangle inequality,
  and~Inequality~\eqref{eq:1} gives
  \[
    d(c_E,C') \le d(c_E,\Mark(c_E)) + d(\Mark(c_E),C') \le 2\offguess +
    \offguess = 3\offguess\,,
  \]
  which proves the first part of the statement of the lemma.  We prove the
  second part next.

  Assume that the promise holds.  Then let $B$ be the set such that cost of $B$
  is at most $\offguess$ with respect to $C_E$, in particular, with respect to
  $C$.  For $c\in C$, define $\Cov(c)\in B$ to be an arbitrarily chosen
  ``coverer'' of $c$, i.e.,
  \begin{equation}
    \label{eq:2}
    d(c,\Cov(c))\le \offguess\,.    
  \end{equation}
  Then the set $B' := \{\Cov(c) : c\in C\}$ is a subset of $B$, so it is
  independent in $\cM$.  Now, for $c, c' \in C$, such that $c\neq c'$, we have
  $\Cov(c)\neq \Cov(c')$ by~Inequality~\eqref{eq:2} because
  $d(c,c') > 2\offguess$.  This implies that $|B'| = |C|$.  Next,
  $\Cov(c) \in B'\cap B_c$ for each $c\in C$, hence the set $B'$ is also
  independent in $\cM_C$.  Therefore $B'\in \cM_C\cap\cM'$, and
  \Call{matroid-intersection}{} returns an independent set of size $|C|$, i.e.,
  it does not fail.
\end{proof}

\begin{remark}
  By running $\binom{|X|}{2}$ guesses, \Call{EfficientMatroidCenter}{} can be
  used to get an offline $3$-approximation algorithm for a more general version
  of matroid center, where the cost is computed with respect to a subset $C_E$
  of $X$ and any point in $X$ can be a center.
\end{remark}
